\documentclass{article}

\usepackage{PRIMEarxiv}

\usepackage{xcolor}
\usepackage[utf8]{inputenc} 
\usepackage[T1]{fontenc}    
\usepackage{hyperref}       
\usepackage{url}            
\usepackage{booktabs}       
\usepackage{amsfonts}       
\usepackage{nicefrac}       
\usepackage{microtype}      
\usepackage{lipsum}
\usepackage{graphicx}
\graphicspath{{media/}}     
\usepackage{amsmath}  
\usepackage{amsthm}
\newtheorem{proposition}{Proposition}

\title{ Notes on Retroactive Interference Model of Forgetting
}

\author{
  Mikhail Katkov \\
	Department of Brain Sciences\\
	Weizmann Institute of Science\\
	Rehovot, 76100 Israel \\ \\
        School of Natural Sciences\\
	Institute for Advanced Study\\
	Princeton, NJ \\
	\texttt{mikhail.katkov@gmail.com} \\
}

\begin{document}
\maketitle


\begin{abstract}
We present analytical derivation of the minimal and maximal number of items retained in recently introduced Retroactive Interference Model of Forgetting. Also we computed the probability that two items presented at different times are retained in the memory at a later time analytically. 
\end{abstract}


\section*{Introduction}

Studying human memory is a complex task, since there are presumably many interacting processes that are hard to isolate. Traditionally models in psychology of memory are attempting to describe many processes at once by creating a very complicated mathematical model that have a lot of parameters, hard to analyse, and usually require a separate fitting of parameters for each experiment. Therefore, since parameters have to adjusted for each measurement, it is not clear how good these models would describe memory processes outside of laboratory settings. We have recently proposed a different type of models to describe human memory that are based on few assumptions and have zero or few parameters that describe the class of stimuli, but independent of experimental settings. These types of models, if validated, have a much broader applicability in everyday life settings. We have recently presented mathematical models for memory forgetting and retrieval \cite{katkov2022mathematical_models}. In this publication we have asked about mathematical properties of forgetting model that may help design experiments to validate the underlying mathematical construction. In this note we provide answers to 2 questions raised in that publication regarding the forgetting model. 

Experimentally, forgetting is traditionally measured in the form of retention function ($RC(t)$) - the probability that memory is retained for time $t$ since acquisition\cite{Ebbinghaus1885}. People observe that retention function monotonically decreasing with time and although there are debates on the form of retention function one of the best candidate is power function of time. One of the popular explanation of memory forgetting in humans is retrograde interference \cite{Wixted2004-R}. It assumes that new incoming memories are interacting with stored memories and cause some past memories disappear. There are different approaches to model forgetting (see for example \cite{finotelli2023mathematical}), but here we are concentrating on consequences of our mathematical model \cite{georgiou2021retroactive} for possible experimental validating of underlying assumptions behind our model. The model assumes that each incoming memory (one memory at one discrete time step) has $n$ dimensional valence vector, with components being iid random variables. Every incoming memory is added to the memory pool, erasing all stored memories that have smaller valences in all dimensions. This model can be solved analytically, and the resulting retention curve agree well with experiment. Nevertheless, it is not clear to what extent the underlying principles holds during retention of memory items. We have posed recently several mathematical questions that may provide additional experimental tests related to this issue\cite{katkov2022mathematical_models}. 

We consider the forgetting model III from \cite{katkov2022mathematical_models}. It states that there is a retention process, where at each time step a new memory is presented to the process. Memory in the model is characterized by a vector of valences $v_k \in \mathbb{R}, k=1..n$, where $n$ is a single integer parameter of the model representing the dimensionality of the model. Valences of incoming memory is assumed to be sampled from arbitrary stationary distribution (absolutely continuous). Incoming memory erases all currently stored memories that have all valences smaller that corresponding valencies of incoming memory. Finally, incoming memory is added to stored memories. Formally, all incoming memories are described by collection of valences $V = \{ v_{k,t} \in \mathbb{R}, k=1..n, t \in \mathbb{N} \}$. For each time $T$ we can define a set of stored memories $M(T; V) = \{t_1, ... t_{|M(T)|}\}, t_k< T$ which contains presentation times ($t_1, ... t_{|M(T)|}$) of stored memories, where the cardinality of $M(T; V)$ ($|M(T; V)|$) represents a number of stored memories at time $T$. At time $T+1$ a new memory with valences $v_{k,T+1}$ is presented, and the set of stored memories is updated $M(T+1; V)= \{T+1; t_m : t_m \in M(T; V), \exists k | v_{k, t_m} > v_{k, T+1} \}$. 
One can ask what is expected value of $|M(T)|$ that is referred to as retention curve $RC_n(T) = E( |M(T; V)| )$ with respect to distribution of $v_{k,t}$. It turns out that retention curve does not depend on the particular distribution of valences, since it depends only on the order statistics of memory items in each dimension\cite{georgiou2021retroactive, katkov2022mathematical_models}. 

\section*{Minimal number of retained items}

We proposed two kind of tests that can potentially check the validity of the model \cite{katkov2022mathematical_models}. Both are related to partialy ordered set (poset) nature of the model. For instance, one can define a partial order relationship between memories by the erasure operation - a memory $A= v_{k,A}$ is larger than $B= v_{k,B}$ if valences in all dimensions are larger $A \succ B \equiv v_{k, A} > v_{k, B} \forall k$. In terms of recall process, it means that larger memory $A$ will erase memory $B$ if it is presented later in time. Due to linear extension principle \cite{mathias1974order_extension}, one can permute presentation order such that all memories will be retained. For instance, permute items such that valences are sorted along dimension $k$ in decreasing order. Consider an item at time $t$, since $v_{k,t_1} < v_{k,t}, \forall t_1 > t $ this item will not be erased. Since this construction is valid for all $t$ all items will be kept in memory in that order. We have asked what would be the minimum number of retained memories across all permutations of presentation order for a given memory realization (a set of $v_{k,t}$). This question is related to the number of maximal elements in the poset. The element of poset is maximal, if there is no element that is greater. In our settings that means that maximal element cannot be erased at any presentation position and the number of retained items cannot be smaller that the number of maximal elements. The first question is related to the distribution of maximal elements for a given poset.

\begin{proposition}
Let $P(|M_n(T; V_{\sigma(T, V)} )| = k )$ denote the probability of retention process $M_n$ of dimension $n$ at time $T$ with presentation time permuted according to $\sigma(T, V)$ has cardinality $k$. Than 

(1) The maximal number of retained items across presentation order permutations is always $T$ for any $V$.

(2) The distribution of minimal number of retained memories $P(|M_n(T; V_{\sigma(T, V)} )| = k ) = P(|M_{n-1}(T; V)| = k )$, where $\sigma(T, V): |M_n(T; V_{\sigma(T, V)} )| \leq  |M_n(T; V_{\sigma'(T, V)} )| \forall \sigma(T, V) $; in other words the distribution of minimal number of retained items in the model with dimension $n$ is the same as distribution of retained items in random presentation order for model with dimension $n-1$. Note that averaging is done over all $V$, whereas for each $V$ there is possibly different permutation $\sigma(T, V)$.
\end{proposition}

\begin{proof}

(1) Sort memories w.l.o.g. (without loss of generality) across dimension 1 in descending order. Then $v_{i,t_2} < v_{i, t_1}, t_1<t_2$ and memory presented at $t_1$ cannot be erased, because valence in all later items is smaller in at least one dimension.

(2) Since the erasure operation is defined with respect to order statistics we need to show that distribution of cardinalities of permuted order statistics corresponds to the distribution of cardinalities for unpermuted memory process of dimensionality one smaller. 

Step 1. Permutation leading to minimal number of retained items.

Sort memories w.l.o.g. (without loss of generality) across dimension 1 in ascending order.
We claim this presentation order is $\sigma(T, V)$. For instance, take any $t \in M_n(T; V_{\sigma(T,V)})$. Any memories with $t'<t$ cannot erase memory $t$ if inserted after memory $t$, since valence in the first dimension is smaller. Memories with $t'>t$ cannot erase memory $t$, because they were presented at a later time in sorted order and did not erase memory $t$. Therefore, each memory $t \in M_n(T; \sigma(T, V))$ will not be erased at any position, i.e. will be present in any temporal presentation. 

Step 2. Permutations in last dimensions

Consider $\mathbb{V}$ a set of all permutations of $V$ along last dimension. For any $V$ the cardinality of $|\mathbb{V}| = T!$ is the same for each V. 
$|M_{n-1}(T; V)| = |M_{n-1}(T; V')| \forall V \in \mathbb{V}$, and there is only one permutation from Step 1 that guarantee non-erasure in last dimension. Therefore, $P(|M_n(T; V_{\sigma'(T, V)} )| = k ) = P(|M_{n-1}(T; V)| = k )$ 
\end{proof}

\section*{Pairwise retention correlations}

Second type of tests for the model comes form the observation in numerical simulations that the correlations between retained items that acquired at different presentation time is non zero for dimensions greater than 1. Here we present analytical derivation of such correlations. Note, that such correlations arise purely from poset nature of the model. It is instructive, first, to derive retention function in slightly different way that was done in \cite{georgiou2021retroactive, katkov2022mathematical_models}. It is essentially the same calculation as in Eq(22) in \cite{katkov2022mathematical_models} written in a different form.

One can describe the forgetting process consisting of 2 Markov processes. One is related to potential erasure of the previous item due to single dimension, and the second one is related to item survival over time step due to first Markov process. In more details, there 2 states in the Markov process over time: `0' - item is erased and `1' item is still present in memory. There are also 2 states in the Markov process over dimensions: `0' - item is potentially erased by following dimension and `1' - item survived time step. For instance, suppose that memory item acquired at time $t'$ survived until time $t$, if $v_{k,t} < v_{k,t'}$ the item will survive independently of valencies relationship in other dimensions. The second Markov process has exactly $n$ steps and the state after the last step determines whether memory survived this time step or not. For example, the state `0' indicates that all valencies of original item are smaller than currently acquired one, and item is erased at this time step if it was in memory in the previous step. 

\subsection*{Retention curve derivation}
Since retention function does not depend on the distribution of valences we can assume that the distribution is uniform between 0 and 1. We first consider what is the probability for the memory with valence $u$ to be retained after $T$ time steps and then average this probability across all realization of $u$

Let
\begin{align}
    P(u; t)= \left( \begin{array}{c}
         P_0 (u; t)  \\
         P_1 (u; t) 
    \end{array} \right)
    \begin{array}{l}
         \mbox{probability the memory is erased after step t}   \\
         \mbox{probability the memory is retained after step t}
    \end{array} , 
\end{align}
where $u$ is n-dim valence of memory. Immediately after presentation the memory is retained, therefore $P(u;0)= (0,1)$. In order to compute this probability vector we need to analyze what happen during a single time step. 

Consider the case when memory is retained at the beginning of a time step. It appears that the process of forgetting in a single time step can be described as Markov process. For instance, the memory can be retained at a time step due to arbitrary single dimension $d$ where the valence of newly presented item is smaller than $u_d$. Therefore, to describe Markov process we need to keep information about probability of the memory retained or potentially erased after each dimension.
Let 
\begin{align}
 \mathcal{P}(u, d) =\left( \begin{array}{c}
        \mathcal{P}_0 (u; d)  \\
         \mathcal{P}_1 (u; d) 
    \end{array} \right)
    \begin{array}{l}
         \mbox{probability memory is potentially erased after dimension d}   \\
         \mbox{probability memory is retained after dimension d }
    \end{array} .
\end{align}

To write the transition matrix due to one dimension one can note that if memory is retained then item is remains retained. However, if if memory is potentially erased it can be retained when valence of newly incoming memory is smaller than $u_d$. The probability of this event is $u_d$. Therefore, we can write transition to the next dimension as
\begin{align}
    \mathcal{P}(u; d)=& \mathcal{E}(u; d) \mathcal{P}(u;d-1) \\
    \mathcal{E}(u; d) =& \left( \begin{array}{cc}
        1-u_d & 0 \\
        u_d & 1
    \end{array}\right), 
\end{align}
where initially memory is potentially erased $\mathcal{P}(u;0)= (1,0)$, and the probability that memory is retained after dimension $n$ is $\mathcal{P}(u;n)$.

One can observe that 
\begin{align}
    \mathcal{E}(u; d)=& \left( \begin{array}{cc}
        0 & -1 \\
        1 & 1
    \end{array}\right) \left( \begin{array}{cc}
        1 & 0 \\
        0 & 1-u_d
    \end{array}\right) \left( \begin{array}{cc}
        1 & 1 \\
        -1 & 0
    \end{array}\right), 
\end{align}
and therefore the probability that memory survives the time step if it was in the memory is
\begin{align}
    \mathcal{P}(u; n)=& \mathcal{E}(u; n) \mathcal{E}(u; n-1) .. \mathcal{E}(u; 1) \left( \begin{array}{c} 1 \\ 0 
    \end{array}\right) \\
    = &  \left( \begin{array}{cc}
        0 & -1 \\
        1 & 1
    \end{array}\right) \left( \begin{array}{cc}
        1 & 0 \\
        0 & \prod \limits_{d=1}^n 1-u_d
    \end{array}\right) \left( \begin{array}{cc}
        1 & 1 \\
        -1 & 0
    \end{array}\right) \left( \begin{array}{c} 1 \\ 0 
    \end{array}\right) \\
    = & \left( \begin{array}{c}  \prod \limits_{d=1}^n 1-u_d \\ 1- \prod \limits_{d=1}^n 1-u_d
    \end{array}\right).
\end{align}

Returning to the first Markov process, one can observe that if the memory was erased before current time step it stays erased and that if the memory was retained before current step it will be forgotten with probability $\prod \limits_{d=1}^n 1-u_d $. Therefore, forgetting process can be written as a Markov process
\begin{align}
    P(u; t+1)=& Er(u) P(u;t) \\
    Er(u) = & \left[ \begin{array}{cc} 
        \left( \begin{array}{c} 1 \\ 0 \end{array}\right) &  
         \mathcal{P}(u; n) 
    \end{array} \right] =
     \left( \begin{array}{cc}
        1 & \prod \limits_{d=1}^n 1-u_d  \\
        0 & 1-\prod \limits_{d=1}^n 1-u_d 
    \end{array}\right), 
\end{align}
where $Er$ is a transition matrix after one time step.

Now we can compute retention process in time
\begin{align}
    P(u; T)=& Er(u) P(u;T-1) \\
    =& Er(u)^T P(u; 0) \\
    =& \left( \begin{array}{cc}
        1 & \prod \limits_{d=1}^n 1-u_d  \\
        0 & 1-\prod \limits_{d=1}^n 1-u_d 
    \end{array}\right)^T \left( \begin{array}{c}  0 \\ 1 \end{array}\right) \\
    = & \left( \begin{array}{cc}
        1 & -1 \\
        0 & 1
    \end{array}\right) \left( \begin{array}{cc}
        1 & 0 \\
        0 & 1- \prod \limits_{d=1}^n 1-u_d
    \end{array}\right)^T \left( \begin{array}{cc}
        1 & 1 \\
        0 & 1
    \end{array}\right) \left( \begin{array}{c}  0 \\ 1 \end{array}\right) \\
    = & \left( \begin{array}{cc}
        1 & -1 \\
        0 & 1
    \end{array}\right) \left( \begin{array}{cc}
        1 & 0 \\
        0 & \left( 1- \prod \limits_{d=1}^n 1-u_d\right)^T 
    \end{array}\right) \left( \begin{array}{cc}
        1 & 1 \\
        0 & 1
    \end{array}\right) \left( \begin{array}{c}  0 \\ 1 \end{array}\right)\\
    = & \left( \begin{array}{c}
        1-\left( 1- \prod \limits_{d=1}^n 1-u_d\right)^T  \\
        \left( 1- \prod \limits_{d=1}^n 1-u_d\right)^T 
    \end{array}\right).
\end{align}
From here 
\begin{align}
    \mbox{RC}_n(T)=& \int \limits_{0}^1 d u_1 ... \int \limits_{0}^1 d u_n  \left(1- \prod \limits_{d=1}^n (1-u_d) \right)^T\\
    =& \int \limits_{0}^1 d u_1 ... \int \limits_{0}^1 d u_n \sum \limits_{m=0}^{T} (-1)^m \binom{T}{m} \prod \limits_{d=1}^n (1-u_d)^m \\
    =& \sum \limits_{m=0}^{T} \binom{T}{m} \frac{(-1)^m}{(m+1)^n}
\end{align}

\subsection*{Retention curve for 2 memories}

We start the process with only one memory with valencies $n$-dimensional $v$. After $t_1$ time steps another memory is presented with valence $w$. After $t$ time steps from this moment we can represent the state of 2 memories as 
\begin{align}
 P(t; u, w)=  \left( \begin{array}{c}
        P_{00} (t; u, w)  \\
        P_{01} (t; u, w) \\
        P_{10} (t; u, w) \\
        P_{11} (t; u, w) 
    \end{array} \right)
    \begin{array}{l}
         \mbox{probability both memories are forgotten}   \\
         \mbox{probability second memory is retained  }   \\
         \mbox{probability first is retained }   \\
         \mbox{probability that both memory are retained after t time steps  }
    \end{array},
\end{align}
We will consider forgetting process as Markov process $P(t+1; u, w)= T_t(u,w) P(t; u, w) $, since each step is independent, and the state of the system depends only on previous step. $T_t(u,w)$ is the transition matrix for single time step that depends on valences of both memories.

At $t=0$ second memory is acquired and since first memory maybe erased after $t_1$ time steps there 2 non-zero values $P_{01}$ and $P_{11}$ with probabilities computed in previous section and due to the presentation of second memory
\begin{align}
 P(0; u, w)=  \left( \begin{array}{c}
         0 \\
        1- \left( 1-  \prod \limits_{i=1}^n  (1- u_i) \right)^{t_1} \left( 1-  \prod \limits_{i=1}^n \Theta( w_i-u_i) \right) \\
         0  \\
        \left( 1-  \prod \limits_{i=1}^n  (1- u_i) \right)^{t_1} \left( 1-  \prod \limits_{i=1}^n \Theta( w_i-u_i) \right)
    \end{array} \right),
\end{align}
where $\Theta( x ) $ is Heaviside theta function.

The transition matrix is 
\begin{align}
 T_t(u,w)=  \left( \begin{array}{cccc}
        1 &  \prod \limits_{i=1}^n  (1- w_i) & \prod \limits_{i=1}^n  (1- v_i)  &  \mathcal{P}_{00} (u, w; n ) \\
        0 &  1-  \prod \limits_{i=1}^n  (1- w_i)  & 0 & \mathcal{P}_{01} (u, w; n)  \\
        0 & 0 & 1- \prod \limits_{i=1}^n  (1- v_i)  &  \mathcal{P}_{10} (u, w; n)  \\
        0 & 0 & 0 & \mathcal{P}_{11} (u, w; n) 
    \end{array} \right),
    \label{eq:T_t_start}
\end{align}
where 
\begin{align}
 \mathcal{P}(u, w; d)=  \left( \begin{array}{c}
        \mathcal{P}_{00} (u, w; d)  \\
         \mathcal{P}_{01} (u, w; d) \\
         \mathcal{P}_{10} (u, w; d) \\
         \mathcal{P}_{11} (u, w; d) 
    \end{array} \right)
    \begin{array}{l}
         \mbox{probability both memory are potentially erased after dimension d}   \\
         \mbox{probability second memory is retained after dimension d }   \\
         \mbox{probability first is retained after dimension d }   \\
         \mbox{probability both memories are retained after dimension d }
    \end{array}.
\end{align}
The first column represents the fact that erased memory stays erased, the second and third columns deals with the processes where only one memory is retained and the last column represent forgetting step when two memories are still present. 

To compute $\mathcal{P}(u, w; d)$ we can consider another Markov process, where the initial state is `both memory are present' and potentially erased by new stimulus. In each dimension the valence of a new stimulus is either smaller than one or both items in which case corresponding memory/ies will survive this time step, or greater, in which case corresponding memory can potentially survive due to next dimension. Therefore, $\mathcal{P}_{00} (u, w; n)$ represent probability that both memories are erased in a single step.

This Markov process can be described as follows
\begin{align}
    \mathcal{P}(u, w; d)=& \mathcal{E}(u, w; d) \mathcal{P}(u, w;d-1) \\
    \mathcal{E}(u, w; d) =& \left( \begin{array}{cccc}
         1- \max(v_d, w_d) & 0 & 0 & 0\\
         (w_d-v_d) \Theta( w_d-v_d) & 1 - v_d  & 0 & 0\\
         (v_d-w_d) \Theta( v_d-w_d) & 0  &  1-w_d & 0\\
        \min(v_d, w_d) & v_d & w_d & 1
    \end{array}\right),\\
    \mathcal{P}(u, w; 0)=  \left( \begin{array}{c}
         1  \\
         0 \\
         0 \\
         0
    \end{array} \right),
\end{align}

The last column states that item that survived previous dimensions will survive this dimension. Previous 2 columns describe the processes when one memory is survived, but second one can potentially be erased, and the first column describe survival process when both memories can potentially be erased.

To simplify analysis we can consider 2 cases: when the first memory has larger valence than the second memory in dimension $d$ $w_d<v_d$ and the opposite case.

\paragraph{First memory has larger valence}
In this case 
\begin{align}
    \mathcal{E}(u, w; d | u_d>w_d ) =& \left( \begin{array}{cccc}
         1-u_d & 0 & 0 & 0\\
         0 & 1 - u_d  & 0 & 0\\
         u_d-w_d & 0  &  1-w_d & 0\\
        v_d & u_d & w_d & 1
    \end{array}\right),\\
\end{align}

Eigenvalues are $1, 1-u_d, 1-u_d, 1-w_d$, with corresponding eigenvectors (as column vectors)
\begin{align}
  U_1 = \begin{pmatrix}
    0 & 0 & -1 & 0 \\ 0 & -1 & 0 & 0 \\ 0 &0 &1&-1\\ 1& 1& 0& 1
    \end{pmatrix}
\end{align}
\paragraph{Second memory has larger valence}
In this case 
\begin{align}
  \mathcal{E}(u, w; d | u_d<w_d ) =& \left( \begin{array}{cccc}
         1-v_d & 0 & 0 & 0\\
         w_d-u_d & 1 - u_d  & 0 & 0\\
         0 & 0  &  1-w_d & 0\\
        w_d & u_d & w_d & 1
    \end{array}\right),\\
\end{align}
Eigenvalues are $1, 1-u_d, 1-w_d, 1-w_d$, with corresponding eigenvectors (as column vectors)
\begin{align}
   U_2 =  \begin{pmatrix}
    0 & 0 & -1 & 0 \\ 0 & -1 & 0 & 0 \\ 0 &0 &1&-1\\ 1& 1& 0& 1
    \end{pmatrix}
\end{align}

It is clear that effect of dimension can be tested in any order, the result would be the same. It can also be checked by direct calculation that $\mathcal{E}(u, w; d_1 | u_{d_1}>w_{d_1} ) \mathcal{E}(u, w; d_2 | u_{d_2}<w_{d_2} ) = \mathcal{E}(u, w; d_2 | u_{d_2}<w_{d_2} ) \mathcal{E}(u, w; d_1 | u_{d_1}>w_{d_1} )$.
Therefore we can apply $\mathcal{E}(u, w; d | u_{d}>w_{d} )$ for all dimensions where $u_{d}>w_{d}$ and then apply  $\mathcal{E}(u, w; d | u_{d}<w_{d} )$ for all dimensions where $u_{d}<w_{d}$

\begin{align}
    \mathcal{P}(u, w; n) = & U_2 \begin{pmatrix} 1 & 0& 0& 0\\ 0& \prod\limits_{d|w_d>u_d} (1-u_d) & 0&0\\
    0&0& \prod\limits_{d|w_d>u_d} (1-w_d)& 0\\ 0&0&0&\prod\limits_{d|w_d>u_d} (1-w_d) \end{pmatrix} U_2^{-1} \times \nonumber \\ & U_1 
     \begin{pmatrix} 1 & 0& 0& 0\\ 0& \prod\limits_{d|w_d<u_d} (1-u_d) & 0&0\\
    0&0& \prod\limits_{d|w_d<u_d} (1-u_d)& 0\\ 0&0&0&\prod\limits_{d|w_d<u_d} (1-w_d) \end{pmatrix} U_1^{-1} 
    \begin{pmatrix} 1 \\ 0 \\ 0 \\ 0    \end{pmatrix}
    \label{eq:prob_over_dim}
\end{align}
Let 
\begin{align}
    F(x) = \prod\limits_{d|w_d<u_d} (1-x_d),\\
    G(x)=  \prod\limits_{d|w_d>u_d} (1-x_d),
\end{align}
than \eqref{eq:prob_over_dim} can be written as
\begin{align}
    \mathcal{P}(u, w; n) = & U_2 \begin{pmatrix} 1 & 0& 0& 0\\ 0& G(u) & 0&0\\
    0&0& G(w)& 0\\ 0&0&0& G(w) \end{pmatrix} U_2^{-1} \times \nonumber \\ & U_1 
     \begin{pmatrix} 1 & 0& 0& 0\\ 0& F(u) & 0&0\\
    0&0& F(u)& 0\\ 0&0&0&F(w) \end{pmatrix} U_1^{-1} 
    \begin{pmatrix} 1 \\ 0 \\ 0 \\ 0    \end{pmatrix}
\end{align}
\begin{align}
    U_1^{-1} \begin{pmatrix} 1 \\ 0 \\ 0 \\ 0    \end{pmatrix}= \begin{pmatrix} 1 \\ 0 \\ -1 \\ -1    \end{pmatrix} 
\end{align}  
\begin{align}
U_1 \begin{pmatrix} 1 & 0& 0& 0\\ 0& F(u) & 0&0 \\
    0&0& F(u)& 0\\ 0&0&0& F(w) \end{pmatrix} 
    \begin{pmatrix} 1 \\ 0 \\ -1 \\ -1    \end{pmatrix} = \begin{pmatrix} F(u) \\0 \\ F(w) -F(u) \\ F(w)  \end{pmatrix}
\end{align}  
\begin{align}
    \mathcal{P}(u, w; n) = \begin{pmatrix}
    F(u)G(w) \\ F(u)(G(u)-G(w)) \\ G(w)(F(w)-F(u)) \\ 1- F(w)G(w) + F(u)G(w)- F(u)G(u)
    \end{pmatrix}
\end{align}

\subsection*{Simplifying Markov process over time}

Transition matrix \eqref{eq:T_t_start} now can be rewritten explicitly (note that)
\begin{align}
    \prod \limits_{i=1}^n  (1- w_i) = F(w)G(w) \\
    \prod \limits_{i=1}^n  (1- u_i) = F(u)G(u) \\
\end{align}

\begin{align}
 T_t(u,w)=  \left( \begin{array}{cccc}
        1 &  F(w)G(w) & F(u)G(u)  &  F(u)G(w) \\
        0 &  1-  F(w)G(w)  & 0 & F(u)(G(u)-G(w))  \\
        0 & 0 & 1- F(u)G(u)  &  G(w)(F(w)-F(u)) \\
        0 & 0 & 0 & 1- F(w)G(w) + F(u)G(w)- F(u)G(u)
    \end{array} \right),
\end{align}

Let further simplify notation $H_{xy} = F(x)G(y)$
\begin{align}
 T_t(u,w)=  \left( \begin{array}{cccc}
        1 &  H_{ww} & H_{uu}  &  H_{uw} \\
        0 &  1-  H_{ww}  & 0 & H_{uu} - H_{uw}  \\
        0 & 0 & 1- H_{uu}  &  H_{ww}-H_{uw}  \\
        0 & 0 & 0 & 1- H_{ww} + H_{uw}- H_{uu} 
    \end{array} \right),
\end{align}

It can be represented as 
\begin{align}
 T_t(u,w)= & U_t \left( \begin{array}{cccc}
        1 & 0 & 0  &  0 \\
        0 &  1-  H_{ww}  & 0 & 0  \\
        0 & 0 & 1- H_{uu}  & 0  \\
        0 & 0 & 0 & 1- H_{ww} + H_{uw}- H_{uu} 
    \end{array} \right) U_t^{-1},\\
    U_t= &\begin{pmatrix} 1&-1&-1&1\\0&0&1&-1\\0&1&0&-1\\0&0&0&1 \end{pmatrix}.
\end{align}

Therefore,
\begin{align}
    P(t, t_1; u, w)=& T_t(u,w) P(t-1, t_1; u, w) = T_t(u,w)^t P(0, t_1; u, w) \\
    =&  U_t \left( \begin{array}{cccc}
        1 & 0 & 0  &  0 \\
        0 &  (1-  H_{ww})^t  & 0 & 0  \\
        0 & 0 & (1- H_{uu})^t  & 0  \\
        0 & 0 & 0 & (1- H_{ww} + H_{uw}- H_{uu} )^t
    \end{array} \right) U_t^{-1} 
    \begin{pmatrix}
     0 \\  1- \left( 1-  Huu \right)^{t_1}  \\
         0  \\    \left( 1-  H_uu \right)^{t_1}
    \end{pmatrix} = \\
    =& U_t \begin{pmatrix}1\\ (1-Huu)^{t_1} \\ (1-Hww)^t\\ (1-Huu)^{t+{t_1}} (1- H_{ww} + H_{uw}- H_{uu} )^t  \end{pmatrix}
\end{align}

\subsection*{Probability of events}
\begin{align}
    P(t, t_1)=& \int \limits_{u, w \in [0,1]^n } du dw P(t, t_1; u, w) \\
    =&  U_t \begin{pmatrix}1\\ \int  du dw (1-Huu)^{t+{t_1}} \\ \int  du dw (1-Hww)^t\\ \int  du dw (1-Huu)^{t_1} (1- H_{ww} + H_{uw}- H_{uu} )^t  \end{pmatrix}\\
    =& U_t \begin{pmatrix}1\\ \sum \limits_{m=0}^{t} \binom{t+t_1}{m} \frac{(-1)^m}{(m+1)^n} \\
    \sum \limits_{m=0}^{t} \binom{t}{m} \frac{(-1)^m}{(m+1)^n}\\
    \int  du dw (1-Huu)^{t_1} (1- H_{ww} + H_{uw}- H_{uu} )^t  \end{pmatrix}
\end{align}
The last integral requires closer attention.
\begin{align}
   P_{11}(t, t_1) = &\int  du dw (1-Huu)^{t_1} (1- H_{ww} + H_{uw}- H_{uu} )^t 
   \label{eq:integral_form}
\end{align}
expanding both brackets into finite series, and then collecting terms
\begin{align}
   = \int  du dw  
    & \left( \sum \limits_{m_1=0}^{t_1} \binom{t_1}{m_1} (-1)^{m_1} Huu^{m_1}\right)
     \left( \sum \limits_{m_2+m_3+m_4+m_5=t} \frac{t!}{m_2!m_3!m_4!m_5!} (-H_{ww})^{m_3} H_{uw}^{m_4} (-H_{uu})^{m_5}\right) \\
     =&  \sum \limits_{m_2+m_3+m_4+m_5=t} \sum \limits_{m_1=0}^{t_1} 
       \binom{t_1}{m_1} \frac{t!}{m_2!m_3!m_4!m_5!} (-1) ^ {m_3+m1+m_5} 
       \int  du dw H_{ww}^{m_3} H_{uw}^{m_4} H_{uu}^{m_1+m_5} .
\end{align}
The integral inside the sum is a product of integrals in each dimension, therefore we can compute integral for each dimension separately. This integral can be computed as a sum of two integrals: one evaluated for triangle where $u_d>w_d$ and another $u_d<w_d$.

{\bf Case 1: $u_d>w_d$}
\begin{align}
    a_1(m)=&\int  \limits_0^1 dw_d (1-w_d)^{m_3} \int  \limits_{w_d} ^1 d u_d   (1-u_d)^{m_1+m_5+m_4} \\
   = & -\frac{1}{(m_1+m_5+m_4+1)} \int  \limits_0^1 dw_d (1-w_d)^{m_3} (1-w_d)^{m_1+m_5+m_4+1} \\
   = &  \frac{1}{(m_1+m_5+m_4+1)}\frac{1}{m_3+m_1+m_5+m_4+2}
\end{align}

{\bf Case 2: $u_d<w_d$}
\begin{align}
    a_2(m)=& \int  \limits_0^1 du_d (1-u_d)^{m1+m_5} \int \limits_{u_d}^1 dw_d  (1-w_d)^{m_3+m_4}  = \\
    =& -\frac{1}{m_3+m_4+1} \int  \limits_0^1 du_d (1-u_d)^{m_1+m_4} (1-u_d)^{m_3+m_5+1} =\\
    = & \frac{1}{m_3+m_4+1} \frac{1}{m_3+m_1+m_5+m_4+2}
\end{align}

Finally, note that in the case when all $w_d > u_d$ the first memory is erased immediately after presentation of second item and corresponding term $a_2^n(m)$ need to be subtracted from computed integral)
\begin{align}
    P_{11}(t, t_1) = &  \sum \limits_{m_2+m_3+m_4+m_5=t} \sum \limits_{m_1=0}^{t_1} 
       \binom{t_1}{m_1} \frac{t!}{m_2!m_3!m_4!m_5!} (-1) ^ {m_3+m1+m_5} \left[ \left( a_1(m)+a_2(m)  \right)^n - a_2(m)^n \right], \nonumber \\
    a_1(m)=&  \frac{1}{(m_1+m_5+m_4+1)}\frac{1}{m_3+m_1+m_5+m_4+2}, \\
    a_2(m)= &\frac{1}{m_3+m_4+1} \frac{1}{m_3+m_1+m_5+m_4+2}. \nonumber
\end{align}

similar correction should be added to all $P_{xy}(t, t_1)$.

We showed a construction for computing probabilities for 2 items retaining in the memory. Similarly, interaction of more memories can be considered. The number of states grows exponentially with the number of memories ($2^M$) and more terms will appear in the integral similar to Eq.~\ref{eq:integral_form}, and the number of cases in computing an integral will grow as $M!$, therefore a computer program can be used to compute probabilities for larger number of memories.

\bibliographystyle{unsrt}  
\bibliography{references}

\begin{thebibliography}{1}

\bibitem{katkov2022mathematical_models}
Mikhail Katkov, Michelangelo Naim, Antonios Georgiou, and Misha Tsodyks.
\newblock Mathematical models of human memory.
\newblock {\em Journal of Mathematical Physics}, 63(7):073303, 2022.

\bibitem{Ebbinghaus1885}
Hermann Ebbinghaus.
\newblock {\em Memory: A contribution to experimental psychology.}
\newblock Dover, New York, 1964.

\bibitem{Wixted2004-R}
John~T Wixted.
\newblock {The psychology and neuroscience of forgetting}.
\newblock {\em Annu. Rev. Psychol}, 55:235--69, 2004.

\bibitem{finotelli2023mathematical}
Paolo Finotelli and Francis Eustache.
\newblock Mathematical modeling of human memory.
\newblock {\em Front. Psychol., Sec. Cognitive Science}, 14, 2023.

\bibitem{georgiou2021retroactive}
Antonios Georgiou, Mikhail Katkov, and Misha Tsodyks.
\newblock Retroactive interference model of forgetting.
\newblock {\em The Journal of Mathematical Neuroscience}, 11(1):1--15, 2021.

\bibitem{mathias1974order_extension}
Adrian Mathias.
\newblock The order extension principle.
\newblock {\em Proceedings of Symposia in Pure Mathematics}, 13(2):179--183,
  1974.

\end{thebibliography}

\end{document}